\documentclass{article}
\usepackage{graphicx, hyperref, amsmath, amsfonts, bbm, geometry, amsthm, amssymb} 

\usepackage[sorting=none, sortcites=true, backend=biber, style=numeric]{biblatex} 
\newtheorem{theorem}{Theorem}
\newtheorem{lemma}{Lemma}
\newtheorem{conjecture}{Conjecture}

\def\E{\mathbb{E}}
\def\V{\text{Var}}

\title{Beyond species-area curves: understanding how diversity relates to area}
\author{Hwai-Ray Tung\thanks{Program in Applied and Computational Mathematics, Princeton University}\and Jérôme Chave\thanks{Centre de Recherche Biodiversité Environnement}\and
Simon A Levin\thanks{High Meadows Environmental Institute (HMEI), Princeton University, United States}\thanks{Department of Ecology and Evolutionary Biology, Princeton University}
}

\begin{document}

\maketitle

\begin{abstract}
    Species-area curves, which describe the number of species present as a function of area, have long been used to understand biodiversity and inform conservation efforts. While species-area curves have been studied for a wide variety of fauna and flora, they are often well described by a power law. Theorists have attempted to explain the power law through many mechanisms, including sampling and spatial effects. While species-area curves are informative, they ignore the population sizes of the various species, which is important when considering ecosystem services. Here, we use mathematically rigorous theory and simulations to predict the relationship between the effective number of species from different diversity indices, like Simpson's index and the Shannon index, and area. Using a spatial Moran model previously used to understand species-area curves, we found that when area is sufficiently large, the relationship between the effective number of species, also termed Hill numbers, and area is linear. On the other hand, the relationship resembles a power law for smaller areas.
\end{abstract}

Biodiversity is a critical consideration for ecosystem services. For example, having high biodiversity can reduce susceptibility to invasive species \cite{hooper2005effects}, which play a major role in 60\% of recorded extinctions and is expected to cost over 423 billion USD in damages annually globally \cite{roy_2024_11629357}. The WHO estimated the global economic costs of declining biodiversity at 10 trillion USD annually \cite{who_biodiversity_2025}. 

A simple way to evaluate biodiversity in a region is to count the number of species present. The relationship between the number of species present and area is called a species-area curve. Since the early 1800s \cite{watson1835remarks}, biologists have recorded species-area curves across a range of geographic areas and multiple forms of flora and fauna, including birds, beetles, mammals, and grasses, among others; see \cite{drakare2006imprint} for a meta analysis of almost 800 species-area curves. These curves have since informed conservation efforts \cite{luby2022and} like wildlife preserve design \cite{may1975islands, desmet2004using, foster2013using} or balancing conservation with economic growth \cite{polasky2008put}. In spite of the wide range of datasets, the relationship between the number of species and area is often well described by a power law, as proposed by Arrhenius in 1921 \cite{arrhenius1921species}; letting ${}^0D$ denote the number of species in area $A$,
\begin{equation}
    {}^0D = cA^z
    \label{eqn:arrhenius}
\end{equation}
where $z$ in data often sits between $0.15$ and $0.4$. 

In search of explanations for the power law relationship, many turned to theoretical models. Preston \cite{preston1962canonical}, for example, argued that if the species present in each area is a random sample of species from the larger population and individuals are distributed across species according to a lognormal law, then a power law with power $z\approx 0.26$ is recovered. Hubbell \cite{hubbell1992speciation, hubbell2011unified} in 1992 was the first to explain using a spatial model. His model takes place on a finite 2D grid, where each site has an identical number of individuals. These individuals may be from different species. At each discrete time step, individuals die with some probability and are replaced with varying probability by either a descendant from the same site, a descendant from one of the eight neighboring sites, or an immigrating individual from a new species. He then used simulations to understand the resulting species-area curves.

This work is strongly inspired by the works of Durrett and Levin \cite{durrett1996spatial} and Bramson, Cox and Durrett \cite{bramson1996spatial}, who studied what can be viewed as a simplified version of Hubbell's model or a variant of what probabilists call the voter model. The primary differences are that the model takes place on an infinite 2D grid, events occur in continuous time, each site has only one individual, and each site has four neighboring sites (up, down, left, and right). These simplifications allowed for \cite{bramson1996spatial} to prove mathematical theorems about the model. \cite{durrett1996spatial} then combined these theorems with simulations to show the model yields a power law species area curve with a range of powers $z$. While both Hubbell's model and the model studied by Bramson, Cox, Durrett, and Levin can be referred to as variants of a spatial Moran model, we use the same model as Bramson, Cox, Durrett, and Levin in this work and for brevity will refer to it as the spatial Moran model in the remainder of this work.

As useful as species-area curves are, they ignore species evenness and paint an incomplete picture of biodiversity; intuitively, the biodiversity of a region with two species, one with a hundred individuals and one with one, and a region also with two species, except both have fifty individuals should be different. While there is no canonical method to quantify diversity, people often use Simpson's index and the Shannon index, which can be used to generate a corresponding effective number of species. These in turn are part of the wider family of indices known as Hill numbers. 

The goal of this work is to understand how the species-area curve changes when using other measures of diversity that include species evenness. The structure is as follows. We start by detailing the spatial Moran model and the different measures of diversity we use. Next, we show simulation results that suggest a relationship resembling a power law between the effective number of species for Simpson's index and the Shannon index for appropriately small areas and a linear relationship for large areas. We then use theorems to estimate the relationship between the power law exponent $z$ and the rate of immigration introducing new species and to prove the linear relationship for sufficiently large areas. We conclude with a brief discussion and leave proofs to the SI Appendix. To our knowledge, there is no empirical or theoretical work to compare our theoretical predictions to. Our hope is that this study will motivate the collection and analysis of appropriate data, as is common for species-area curves.

\subsection*{Model}

\subsubsection*{A spatial Moran model} 
We use the same model as in \cite{durrett1996spatial, bramson1996spatial}, which probabilists may call multitype voter model with mutation and ecologists may call a spatial neutral Moran model with immigration. For brevity, we will refer to it as the spatial Moran model in the remainder of this work. Mathematically, the process is that for each site on the 2D lattice:
\begin{itemize}
    \item At rate $1$, change species to the species of one of the four neighboring spots chosen uniformly at random.
    \item At rate $\alpha$ turn into a new species.
\end{itemize}
The first can be interpreted as an individual reproducing and replacing a neighbor, and the second can be interpreted as an individual being replaced by a brand new species, which one might think of as an immigrant from far away (eg a bird that just flew to an island) or a mutation in a bacterial population that has changed the species. We assume $\alpha \ll 1$, i.e. the introduction of new species is rare relative to reproduction. We study the process at its stationary distribution, which exists by a short argument in \cite{bramson1996spatial}. This removes any dependence on an initial state. 

We seek the Simpson's and Shannon index for an $M\times M$ subset of the grid for varying $M$, making land area $A=M^2$. As we will work predominantly with log-log plots, we choose the convention of setting $M = L^r$ where $L = 1/\sqrt{\alpha} \gg 1$ is a natural length scale, and we vary $r$ rather than varying $M$ directly. This makes $r = \ln(M)/\ln(L)$ a measure of distances in log scale and makes the canonical power law species-area curve \eqref{eqn:arrhenius}
$$
{}^0D=CA^z = CM^{2z} = CL^{2rz}
$$
or in log scale for large $L$,
\begin{equation}
    \ln({}^0D) \approx 2rz\ln(L)
    \label{eqn:loglogArrhenius}
\end{equation}
The definition of $L$ is explained by duality with coalescing random walks. Figure \ref{fig:model_image} provides a depiction of the model. 

\begin{figure}[ht]
    \centering
    \includegraphics[width=0.5\linewidth]{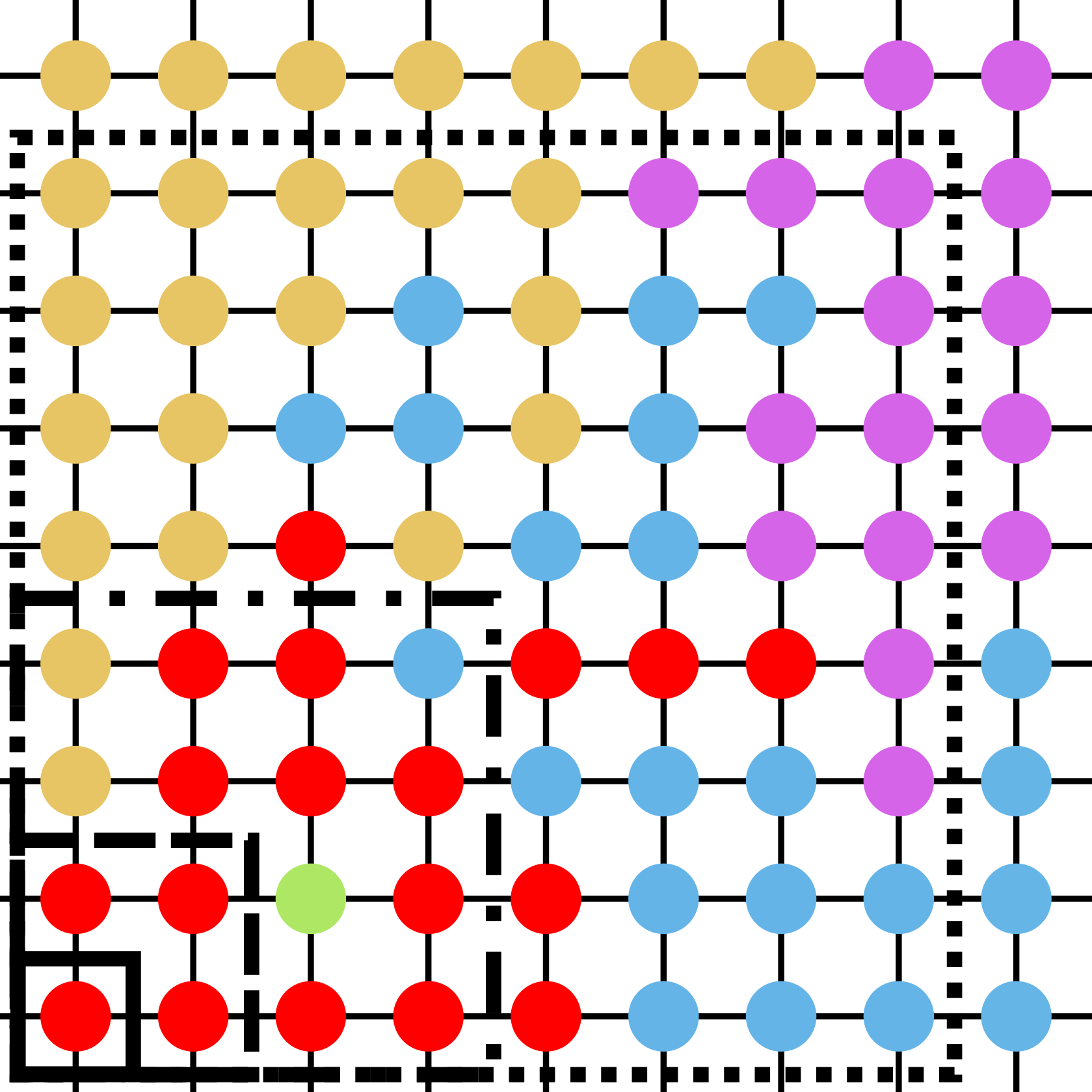}
    \caption{\textbf{Depiction of the spatial Moran model.} Each site on the infinite 2D lattice is occupied by a single individual, which belongs to a species denoted by its color. Over time, individuals die and their site will be replaced by a uniform randomly chosen neighbor's species or a previously unseen species by immigration or mutation. When evaluating diversity measures over area, we take nested regions of squares, as denoted in the figure by the $1\times1$, $2\times2$, $4\times4$, and $8\times8$ boxes. The box side lengths increase exponentially as we will be taking log-log plots later.}
    \label{fig:model_image}
\end{figure}

\subsubsection*{Duality with coalescing random walks}
One of the features that makes the the spatial Moran model mathematically convenient to analyze is its duality with coalescing random walks (CRW). As the name implies, the CRW model starts with many random walkers that, when two walkers meet, coalesce or fuse into one walker. We remove a walker when a mutation occurs. As we illustrate in Figure \ref{fig:duality}, the idea is that running the spatial Moran model backwards in time yields a CRW model. This concept is referred to in probability theory as the duality between the voter model and coalescing random walks, and a more careful treatment can be found in chapter 7 of \cite{Durrett_2026}.

\begin{figure}[ht]
    \centering
    \includegraphics[width=0.65\linewidth]{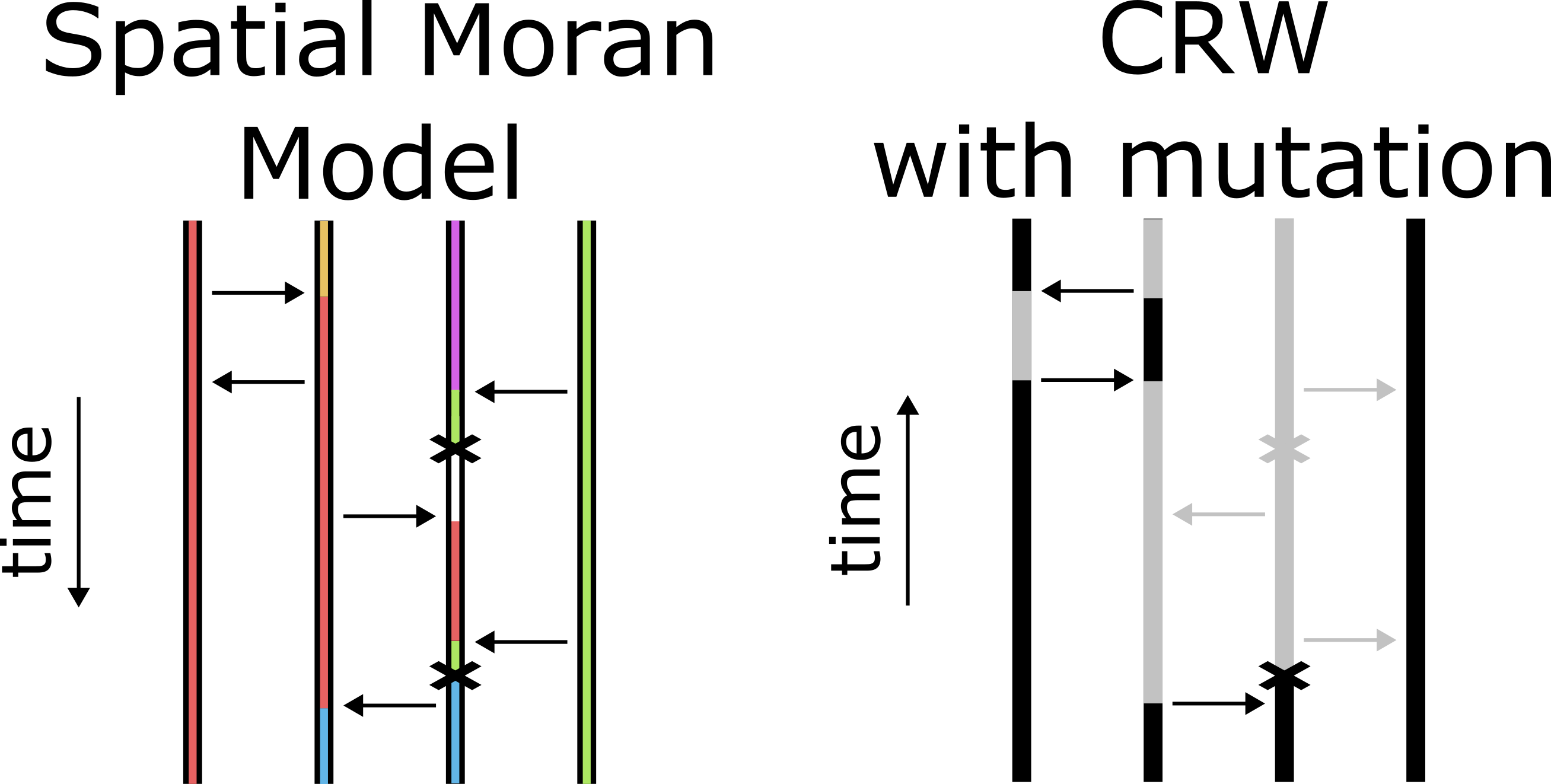}
    \caption{\textbf{The connection between the spatial Moran model and coalescing random walks.} The top row depicts a graphical representation realization of the spatial Moran model and its corresponding CRW on four consecutive sites of the 1D line. The spatial Moran starts with four individuals, each of a different species (colored red, yellow, magenta, and green) at the top of the four lines, which represent the four sites through time. As time progresses, individuals replace their neighbors with their offspring, as denoted by the arrows. Individuals may also be replaced by a mutant or immigrant from a novel species, as denoted by the crosses. At the end of the process, the red species is present on the leftmost site, the blue species occupies the middle two sites, the green species occupies the rightmost site, the originally present yellow and magenta species have gone extinct, and the white species introduced by mutation/immigration has also gone extinct. For the corresponding CRW, we flip the direction of the arrows and start from the bottom of the lines. Tracing the walks back through time reveals the ancestral lineage of each site. Note that two sites have the same species at the bottom if the random walks of their ancestral lineages meet before a mutation/immigration event occurs.}
    \label{fig:duality}
\end{figure}

This duality can be used to explain why a species in the spatial Moran model is likely to be contained in a box with side length of order $L$. In order for two sites to be the same species, their ancestral lineages must meet before a mutation occurs. As the ancestral lineages are random walks, they travel around distance $\sqrt{t}$ by time $t$. Since mutations occur at rate $\alpha$, mutations will occur at around time $1/\alpha = L^2$. Putting this together, the probability two sites are the same drops sharply when they are distance greater than order $L$ apart. This is also why we define $L=1/\sqrt{\alpha}$ as a natural length scale.

\subsubsection*{Diversity and effective number of species} 
Rather than focus on the number of species ${}^0D$, we seek to understand the Simpson's $\sigma$ and Shannon $H$ indices over an $M\times M$ subset of the grid. They are defined as
$$
\sigma = \sum_i p_i^2, \qquad H = -\sum_i p_i \ln(p_i)
$$
where $p_i$ is the fraction of individuals who belong to species $i$ in the set of interest (for us, the $M\times M$ subset). $\sigma$ can be interpreted as the probability two randomly chosen individuals are the same species and is used in many fields, for example as the Herfindahl index in economics; $H$ is a common measure of entropy used in thermodynamics and information theory. In the least diverse case (all sites belong to the same species), $\sigma=1$ and $H = 0$. In the most diverse case (all sites belong to different species), $\sigma = 1/M^2$ and $H =2 \ln(M)$. As such, $\sigma$ ranges between $0$ and $1$ and decreases with diversity, while $H$ ranges between $0$ and $\infty$ and increases with diversity.

To compare ${}^0D, \sigma,$ and $H$, we will also look at their effective number of species using Hill numbers. The $q$th Hill number is defined as
$$
{}^qD = \left(\sum_i p_i^{q}\right)^{1/(1-q)}
$$
and the effective number of species for $H$ and $\sigma$ are ${}^1D$ and ${}^2D$, respectively. We note ${}^0D$ is consistent with the number of species. Since the equation is not defined for $q=1$, we take ${}^1D$ to be the limit of ${}^qD$ as $q\rightarrow 1$. We can then confirm
\begin{equation}
    {}^1D = e^H,\qquad {}^2D = 1/\sigma
    \label{eqn:ENoS}
\end{equation}
We give examples of computing $H, \sigma,$ and ${}^qD$ in Figure \ref{fig:diversity_calcs}.

\begin{figure}[ht]
    \centering
    \includegraphics[width=0.7\linewidth]{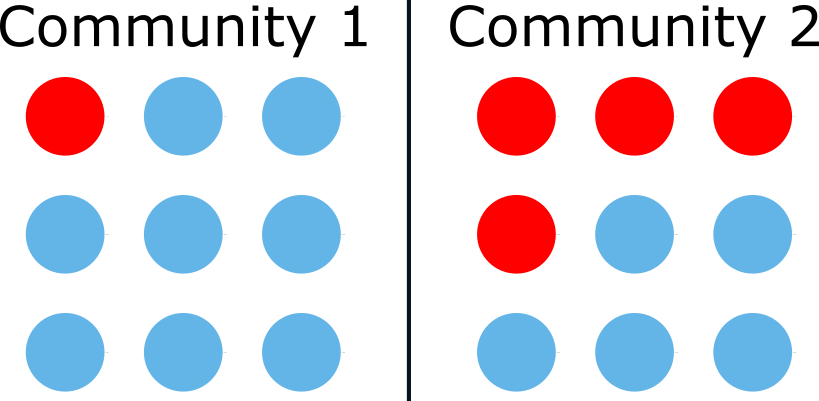}
    \begin{tabular}{|c|c|c|}
    \hline
    & Population 1 & Population 2  \\ \hline
    $H$& $-\frac{1}{9}\ln(\frac{1}{9})-\frac{8}{9}\ln(\frac{8}{9}) \approx 0.35$ &  $-\frac{4}{9}\ln(\frac{4}{9})-\frac{5}{9}\ln(\frac{5}{9})\approx 0.69$ \\ \hline
    $\sigma$& $(\frac{1}{9})^2+(\frac{8}{9})^2 \approx 0.80$ & $(\frac{4}{9})^2+(\frac{5}{9})^2 \approx 0.51$ \\ \hline
    ${}^0D$& 2 & 2  \\ \hline
    ${}^1D$& $e^H\approx 1.42$ & $e^H\approx 1.99$ \\ \hline
    ${}^2D$& $1/\sigma\approx 1.24$ &  $1/\sigma\approx 1.98$ \\ \hline
    \end{tabular}
    \caption{\textbf{Example calculations of diversity measures in two communities.} The graphic depicts two example communities with two species. Community 1 has one individual belonging to the red species and eight individuals belonging to the blue species. Community 2 has four individuals belonging to the red species and five individuals belonging to the blue species. The table depicts different diversity measures for both communities. Although both communities have the same number of species, community 2 has a much higher ${}^1D$ and ${}^2D$ because of its more even split.}
    \label{fig:diversity_calcs}
\end{figure}

Just as $\sigma$ and $H$ are used outside of ecology, the same is true for Hill numbers. For instance, ${}^2D$ is the Laakso-Taagepera index in political science, and $\ln({}^qD)$ is Rényi entropy in physics.

\subsection*{Simulation and theory}
In this section, we provide simulation results, then explain how theory predicts these simulations.

\subsubsection*{Simulations}
We use the CRW model to sample the stationary distribution of the spatial Moran model for $L \in \{100, 400, 1600\}$ and various intervals of $r$. For each sample, we examine the sites in a $\lfloor L^{r}\rfloor \times \lfloor L^{r}\rfloor$ box from the origin and compute diversity indices ${}^0D, \sigma, $ and $H$. This lets us compute average diversity indices for each $L$ and $r$ pair. To generate the effective numbers of species, we plug the diversity indices into \eqref{eqn:ENoS}. Our plots use the axes $\ln({}^qD)/(2\ln(L))$ versus the effective $r$ value, which is $\log_L{\lfloor L^{r}\rfloor}$. This is a normalized log-log plot of effective number of species ${}^qD$ versus the area $L^{2r}$. The results are in Figure \ref{fig:sim_plots}. In the subsequent subsections, we discuss the ordering in size between the three effective numbers of species, linear behavior when $r>1$, and a power-law-like behavior on average when $r<1$.

\begin{figure}[h!]
    \centering
    \includegraphics{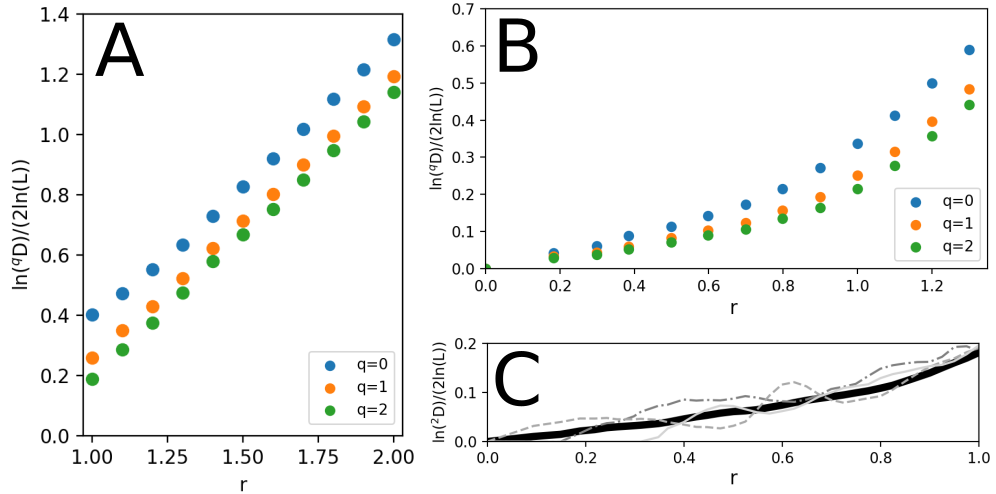}
    \caption{\textbf{Normalized log-log plots of the effective numbers of species versus area.} (\textbf{A}) normalized log-log plot of Hill numbers ${}^0D, {}^1D,$ and ${}^2D$ versus area for a single sample where $L=100$ and $1\leq r \leq 2$. As noted in \eqref{eqn:ENoS}, ${}^0D, {}^1D,$ and ${}^2D$ correspond to the number of species, the effective number of species for Shannon index, and the effective numbers of species for Simpson's index, respectively. (\textbf{B}) a similar normalized log-log plot, except it is the average of 20 samples when $L=400$ and $0\leq r \leq 1.3$. (\textbf{C}) normalized log-log plot of the effective number of species for Simpson's index for three samples (thin grey lines) and for the average of 20 samples (thick black line) when $L=1600$ and $0\leq r \leq 1$. The simulations are run as described in the section ``Simulations''.}
    \label{fig:sim_plots}
\end{figure}

\subsubsection*{Ordering of the effective number of species}
In Figure \ref{fig:sim_plots}AB, the number of species is always the largest, followed by the effective number of species for Shannon and lastly Simpson, i.e. ${}^0D \geq {}^1D \geq {}^2D$. Intuitively, the number of species is the largest because both Simpson's and Shannon give less weight to species with smaller populations; referring back to the example in Figure \ref{fig:diversity_calcs}, the species in community 1 with one individual contributes only $1/9^2\approx 0.01$ to $\sigma$ and barely changes ${}^2D$ but increases the number of species ${}^0D$ by 1.

Mathematically, Hill noted ${}^0D \geq {}^1D \geq {}^2D$ is equivalent to an arithmetic mean-geometric mean-harmonic mean (AM-GM-HM) inequality \cite{hill1973diversity}. More generally, it is known that the ordering holds and ${}^qD$ decreases with increasing $q$; if $0 \leq q_1 < q_2$ and we define $q_{i*} = q_i-1$, then
\begin{equation}
    {}^{q_1}D = \bigg(\sum_i p_i^{q_1}\bigg)^{1/(1-q_1)} = \frac{1}{\left(\sum_i p_i p_i^{q_{1*}}\right)^{1/q_{1*}}} \geq \frac{1}{\left(\sum_i p_i p_i^{q_{2*}}\right)^{1/q_{2*}}} = \bigg(\sum_i p_i^{q_2}\bigg)^{1/(1-q_2)} = {}^{q_2}D
     \label{eqn:hillNumOrder}
\end{equation}
where the inequality in the middle follows from the weighted power mean inequality, which is a special case of Jensen's inequality. We have equality if and only if $p_i = 1/{}^0D$ for all $i$.

\subsubsection*{Linear behavior for $r>1$}
As seen in Figure \ref{fig:sim_plots}A, the slopes of the curves approach $1$. For the number of species, Theorem 1 in \cite{bramson1996spatial}, which uses $N_r$ to denote the number of species rather than ${}^0D$, shows that as $\alpha \rightarrow 0$,
\begin{equation}
    \frac{\ln({}^0D)}{2\ln(L)} \rightarrow (r-1)^+=\max(r-1, 0)
    \label{eqn:rBigger1richness}
\end{equation}
in probability, which tells us when $r>1$
$$
\frac{\ln({}^0D)}{2\ln(L)} \approx r-1
$$
Recalling that the side of length of the box is $M=L^r$, then $r>1$ is equivalent to $M>L$, and 
$$
\ln({}^0D) \approx 2(r-1)\ln(L) = \ln(M^2/L^2)\Rightarrow {}^0D\approx M^2/L^2
$$ 
Since the area of the box $A=M^2$, the number of species therefore increases linearly with area.

We show analogous results for the effective numbers of species for Simpson's index and Shannon index from the more general theorem for any of the Hill numbers,
\begin{theorem}
    For $q\geq 0$, as $\alpha \rightarrow 0$
    $$
    \frac{\ln({}^qD)}{2\ln(L)} \rightarrow (r-1)^+
    $$
    in probability.\label{thm:rBigger1}
\end{theorem}
We leave the proof for the SI Appendix. The intuition for the linear behavior goes back to the duality with CRWs. Since random walks travel a distance of around $L$ before mutation, sites that are greater than distances of order $L$ are very unlikely to affect each other and effectively independent. As such, all effective numbers of species, not only the number of species, scale linearly when $r>1$.

Given that power law exponents for species-area curves in the literature are typically between 0.15 and 0.4, one might come to the conclusion that linear behavior (which is a power law with an exponent of $1$) is not biologically reasonable; this is a half truth. The 0.15 to 0.4 range is for regions with intermediate sized areas. For larger regions on the scale of a continent, biologists have observed near linear relationships \cite{preston1960time, shmida1985biological, hubbell2011unified}, matching the results of \cite{durrett1996spatial, bramson1996spatial}. Nevertheless, operating on the continental scale is not relevant to many conservation problems. As such, we instead turn to the regime where $r<1$.

\subsubsection*{Power-law-like behavior for $r<1$}
When $r<1$, species are highly likely to extend beyond the $L^r \times L^r$ box of interest, making theoretical work far more difficult than in the $r>1$ case. Instead, we use the same approach in \cite{durrett1996spatial} and note that the average ${}^qD$ curve when $q\in \{0, 1, 2\}$ and $0\leq r \leq 1$ in Figure \ref{fig:sim_plots}BC are close to a straight line that goes through the origin. This suggests that, if we let $z_q$ denote the slope for the $q$th Hill number ${}^qD$,
$$
\frac{\ln({}^qD)}{2\ln(L)} \approx z_qr
$$
Multiplying $2\ln(L)$ to both sides and recalling the area of the region is $M^2 = L^{2r}$, we find
$$
\ln({}^qD) \approx 2z_qr\ln(L) = \ln(L^{2z_qr})=\ln(A^{z_q})
$$
Exponentiating both sides, we conclude the relationship between the effective number of species and area is roughly a power law with exponent equal to the slope of the curve, i.e. 
$$
{}^qD \approx A^{z_q}
$$
with $z_q$ coinciding with $z$ in \eqref{eqn:arrhenius}. Furthermore, $z_q$ is observed in Figure \ref{fig:sim_plots} to change with $L$, yielding a range of powers. We now generate estimates of $z_q$ with respect to $L$.

To estimate the slope of the line, we need two points. One is immediate; when $r=0$ (or $M=1$), we are looking at only one site, and therefore the effective number of species is $1$. When considering the number of species, \cite{durrett1996spatial} used Theorem 2 from \cite{bramson1996spatial}, which uses $N_r$ to denote ${}^0D$. They estimated that when $r=1$ (or $M=L$), then ${}^0D \approx \frac{2}{\pi}(\ln(L))^2$. This yields a slope estimate of the log number of species vs log area plot, which is an estimate of the power law exponent, of
$$
z_0=\frac{\ln(\frac{2}{\pi}(\ln(L))^2) -\ln(1)}{\ln(L^2)-\ln(1^2)}=\frac{2\ln(\ln(L)) + \ln(2/\pi)}{2\ln(L)}
$$

We similarly estimate the power law exponent for other effective numbers of species by estimating ${}^qD$ when $r=1$. For $q=2$, we start by estimating Simpson's index.
\begin{theorem}
    Let $X$ denote the sites on the infinite lattice that have the same species as the origin and let $r>1$. Then as $L\rightarrow \infty$
    \begin{equation}
        L^{2r}\frac{\E[\sigma]}{\E[|X|]} \rightarrow 1
        \label{eqn:ES}
    \end{equation}
    and
    \begin{equation}
        \frac{\sigma}{\E[\sigma]} \rightarrow 1
        \label{eqn:Sprob}
    \end{equation}
    in probability.
    \label{thm:Simpson}
\end{theorem}

By \cite{sudbury1976size, sawyer1977rates}, $\E[|X|] \approx \pi L^2/(2\ln(L))$, so based on Theorem \ref{thm:Simpson} we might guess that when $r=1$,
$$
\sigma \approx \frac{\pi}{2\ln(L)}
$$
Using \eqref{eqn:ENoS} and repeating the same steps we did for the number of species yields the estimate of the power law exponent
$$
z_2=\frac{\ln(\frac{2}{\pi}\ln(L)) -\ln(1)}{\ln(L^2)-\ln(1^2)}=\frac{\ln(\ln(L)) + \ln(2/\pi)}{2\ln(L)}
$$

The results of Theorem \ref{thm:Simpson} come from observing that $\E[\sigma]$ is the expected fraction of sites that share the same species as a uniform randomly chosen site. One can similarly interpret the second moment of $\sigma$. Results from \cite{sawyer1979limit} regarding moments of species sizes then yield the desired results. For details on the proof as well as a corresponding conjecture and power law exponent estimate on $H$ and ${}^1D$, see the SI Appendix.

\subsection*{Discussion}
In this work, we examine how the effective number of species for various diversity indices relate to area in a spatial neutral Moran model with immigration on a 2D square lattice. In the regime where $r>1$, we prove that any of the Hill numbers grow linearly with respect to area. In the regime where $r<1$, simulations suggest a relationship similar to a power law between the various effective numbers of species and area. For Simpson's index, we prove a convergence in probability result that enables estimates of the power in the power law relationship.

The model is simple and ignores several factors that almost certainly contribute to the relationship between diversity and area. For example, geographic differences create regions preferable to one species compared to others or regions that are more remote and have decreased immigration. Even without geographical differences, assuming individuals from different species have identical fitness and interact only through competition for space is likely inappropriate in some scenarios. While it may be tempting to incorporate these details, doing so would make the model specific to the implementation of the details and remove all chances of analytic tractability.

The analysis done here could be repeated for other Hill numbers and diversity indices. For any specific Hill number ${}^qD$, we echo Hill in advising that ``Indices non sunt multiplicandi praeter necessitatem... the use of diversity number of `peculiar' orders... is to be strongly discouraged'' \cite{hill1973diversity}. Given that a Hill number is a summary statistic, choosing values of $q$ that make ${}^qD$ interpretable is important. That said, it may be valuable to understand the curve relating ${}^qD$ to $q$ and how this curve changes with area.

One might wonder whether a power law is the correct relationship when $r<1$. After all, there appears to be some curvature around $r=1$, individual samples may deviate away from the average, and there are many functions which look linear for an interval on a log-log plot. Nevertheless, the power law is a natural starting point of investigation. In future work, we will test our predicted relationship between diversity and area with data; we hope others will also collect data and assess the empirical relationship for diversity and area, as is common for species-area curves. The fact that the spatial Moran model correctly reproduces the power law relationship for species-area curves is heartening and suggests the model may be useful for understanding other features of biodiversity.

\section*{Acknowledgments}
HRT was supported by a gift from William H. Miller III to Princeton University. ChatGPT 5 was used to convert Python code to C++ code. Both code files can be found at \href{https://github.com/hwairaytung/Hill-Number-Area-Curves}{here}.

\section*{Supplemental information}
\subsection*{Proof of Theorem \ref{thm:rBigger1}}
By \eqref{eqn:hillNumOrder} and \eqref{eqn:rBigger1richness}, for any $q$ we have 
$$
\frac{\ln({}^qD)}{2\ln(L)} \leq \frac{\ln({}^0D)}{2\ln(L)}\rightarrow (r-1)^+
$$
as $\alpha \rightarrow 0$. As such, it suffices to show the theorem holds for any integer $q>1$. As ${}^qD \geq 1$, the theorem holds for $r\leq 1$ and it suffices to find an appropriate lower bound for ${}^qD$ for $r>1$. We therefore aim to show that for $r>1$ and for any $0<\epsilon< r-1$,
$$
\lim_{L\rightarrow \infty}P\left(r-1-\frac{\ln({}^qD)}{2 \ln(L)} > \epsilon\right) = 0
$$

Letting $q_* = q-1>0$ and using the definition of Hill numbers, we find
\begin{align*}
    r-1-\frac{\ln({}^qD)}{2 \ln(L)} > \epsilon \quad &\Leftrightarrow \quad \frac{\ln(1/{}^qD)}{2 \ln(L)} > \epsilon -r+1 \\
    &\Leftrightarrow \quad 1/{}^qD > L^{2(\epsilon -r+1)} \\
    &\Leftrightarrow \quad \left(\sum_i p_i^{q}\right)^{1/(q-1)} > L^{2(\epsilon -r+1)} \\
    &\Leftrightarrow \quad \sum_i p_i p_i^{q_*} > L^{2(\epsilon -r+1)q_*}
\end{align*}
It then follows by the Markov inequality that
$$
P\left(r-1-\frac{\ln({}^qD)}{2 \ln(L)} > \epsilon\right) = P\left(\sum_i p_i p_i^{q_*} > L^{2(\epsilon -r+1)q_*}\right) \leq \frac{\E[\sum_i p_i p_i^{q_*}]}{L^{2(\epsilon -r+1)q_*}}
$$

Note that we can interpret $\sum_i p_i p_i^{q_*}$ as the probability that for a uniform randomly chosen site $x$ in our $L^r\times L^r$ box of interest $B$, the next $q_*$ uniform randomly chosen sites in the box of interest belong to the same species as $x$. Letting $X$ denote the sites on $\mathbbm{Z}^2$ that share the same species as $x$, it then follows that
$$
\E[\sum_i p_i p_i^{q_*}] = \E\left[\left(\frac{|X\cap B|}{L^{2r}}\right)^{q_*}\right] \leq \frac{\E[|X|^{q_*}]}{L^{2rq_*}}
$$
Plugging back yields
$$
P\left(r-1-\frac{\ln({}^qD)}{2 \ln(L)} > \epsilon\right) \leq \frac{\E[|X|^{q_*}]}{L^{2(\epsilon+1)q_*}} = \E\left[\left(\frac{|X|}{\E[|X|]}\right)^{q_*}\right]\frac{\E[|X|]^{q_*}}{L^{2(\epsilon+1)q_*}}
$$
As $L\rightarrow \infty$, Saywer \cite{sawyer1979limit} showed the moments of $|X|/\E[|X|]$ converge to the moments of a unit exponential random variable, and \cite{sudbury1976size, sawyer1977rates} show that $\E[|X|]$ is on the order of $L^2/\ln(L)$. Therefore, 
$$
\lim_{L\rightarrow \infty}P\left(r-1-\frac{\ln({}^qD)}{2 \ln(L)} > \epsilon\right) =0
$$
which proves the desired result.

\subsection*{Proof of Theorem \ref{thm:Simpson}}
We start by showing a Lemma about how widely dispersed a species can be, then show \eqref{eqn:ES} and lastly \eqref{eqn:Sprob}.

Let $x$ denote a uniformly randomly chosen site in our $M\times M$ box of interest, which we refer to as $B$. Let $X$ denote the sites on the infinite lattice that are the same species as $x$. Then
\begin{lemma}
    Let $B_x$ denote the sites within $L^{1+\epsilon}$ steps from $x$ for $0<\epsilon<1$. Then
    \begin{equation}
        \lim_{L\rightarrow \infty}P(X \subseteq B_x) = 1
        \label{eqn:inRadiusBox}
    \end{equation}

    This implies that when $r>1$,
    \begin{equation}
        \lim_{L\rightarrow \infty}P(X \subseteq B) = 1
        \label{eqn:inBox}
    \end{equation}
    \label{lem:inBox}
\end{lemma}
\begin{proof}
    We start with \eqref{eqn:inRadiusBox}. Note by Markov inequality,
    $$
    P(X\nsubseteq B_x) = P(|X\backslash B_x| \geq 1) \leq \E[|X\backslash B_x|]
    $$
    We now use the argument given in \cite{sudbury1976size} to relate $\E[|X\backslash B_x|]$ to the range of a random walk, which is the number of distinct sites a random walk reaches before some time $t$. $\E[|X\backslash B_x|]$ is the sum over sites $x^*$ outside of $B_x$ of the probability rate $1$ random walks from $x$ and $x^*$ meet before one of them mutates, each at rate $\alpha$. Subtracting the random walks from each other and rescaling time, we find this is the same as summing over the probability the rate $1$ random walk from $x$ reaches $x^*$ before a mutation at rate $\alpha$. This is the same as the expected number of unique sites $Y$ outside of $B_x$ that a rate 1 random walk starting $x$ reaches before mutation at rate $\alpha$. Then
    $$
    \E[|X\backslash B_x|] = \E[|Y\backslash B_x|] = \E[|Y\backslash B_x|\mathbbm{1}_{Y \nsubseteq B_x}] \leq \sqrt{\E[|Y\backslash B_x|^2]P(Y \nsubseteq B_x)}
    $$
    by Holder's inequality.

    A cheap upper bound for $\E[|Y\backslash B_x|^2]$ is that $|Y\backslash B_x|$ must be less than one plus the number of steps taken by the random walk before mutation. Since the number of steps is geometric with success probability $\alpha/(1+\alpha)$, we use the second moment of the geometric to conclude that $\E[|Y\backslash B_x|^2]$ grows at or slower than $L^4$.

    To upper bound $P(Y \nsubseteq B_x)$, we use the old trick of thinking of a random walk on $\mathbbm{Z}^2$ as the sum of two 1D random walks, one with steps $\pm(1/2, 1/2)$ and the other with steps $\pm(-1/2, 1/2)$. $Y \nsubseteq B_x$ if one of the 1D walks travels a distance greater than $L^{1+\epsilon}$, which by Azuma-Hoeffding is less than $2\exp(-L^{2+2\epsilon}/(2k))$ where $k$ is the number of steps taken. It then follows that
    \begin{align*}
        P(Y \nsubseteq B_x) &\leq \sum_{k=1}^\infty \frac{\alpha}{1+\alpha} \left(\frac{1}{1+\alpha}\right)^k 4e^{-\frac{L^{2+2\epsilon}}{2k}}\\
        &\leq \sum_{k=1}^{L^{2+\epsilon}} \frac{\alpha}{1+\alpha} \left(\frac{1}{1+\alpha}\right)^k 4e^{-\frac{L^{2+2\epsilon}}{2k}} + \sum_{k=L^{2+\epsilon}}^\infty \frac{\alpha}{1+\alpha} \left(\frac{1}{1+\alpha}\right)^k 4e^{-\frac{L^{2+2\epsilon}}{2k}}
    \end{align*}
    The first term we upper bound by ignoring the $1/(1+\alpha)^k$ term and then bounding each term with setting $k=L^{2+\epsilon}$. This yields
    $$
    \sum_{k=1}^{L^{2+\epsilon}} \frac{\alpha}{1+\alpha} \left(\frac{1}{1+\alpha}\right)^k 4e^{-\frac{L^{2+2\epsilon}}{2k}} \leq \sum_{k=1}^{L^{2+\epsilon}} \alpha  4e^{-\frac{L^{2+2\epsilon}}{2k}}\leq L^{2+\epsilon} \alpha  4e^{-\frac{L^{\epsilon}}{2}} \leq 4 L^{\epsilon}e^{-\frac{L^{\epsilon}}{2}}
    $$
    The second term we upper bound by ignoring the exponential and using the survival function of a geometric distribution to find
    $$
    \sum_{k=L^{2+\epsilon}}^\infty \frac{\alpha}{1+\alpha} \left(\frac{1}{1+\alpha}\right)^k 4e^{-\frac{L^{2+2\epsilon}}{2k}} \leq 4\sum_{k=L^{2+\epsilon}}^\infty \frac{\alpha}{1+\alpha} \left(\frac{1}{1+\alpha}\right)^k \leq  4\left(\frac{1}{1+\alpha}\right)^{L^{2+\epsilon}} =  4\left(\frac{L^2}{1+L^2}\right)^{L^{2+\epsilon}}
    $$
    which shrinks with order $e^{-L^{\epsilon}}$ as $L\rightarrow \infty$. As such, $\sqrt{\E[|Y\backslash B_x|^2]P(Y \nsubseteq B_x)} \rightarrow 0$, which implies 
    \begin{equation}
        \E[|X\backslash B_x|] \rightarrow 0
        \label{eqn:expectedOutsideBox0}
    \end{equation}
    and therefore $P(X\nsubseteq B_x) \rightarrow 0$ and $P(X\subseteq B_x)\rightarrow 1$, showing \eqref{eqn:inRadiusBox}.

    We now show \eqref{eqn:inBox} from \eqref{eqn:inRadiusBox}. Choose $\epsilon$ such that $1+\epsilon < r$. Then for a uniform randomly chosen $x$ in $B$, $P(B_x \subset B) \rightarrow 1$. As $X \subseteq B_x$ and $B_x \subset B$ are independent, it follows that $P(X \subseteq B)\geq P(X \subseteq B_x)P(B_x \subset B) \rightarrow 1$.
\end{proof}

We will now prove \eqref{eqn:ES}. Note that $\E[\sigma]$ can be interpreted as the probability two uniformly randomly selected sites in $B$ are the same species. Therefore
$$
\E[\sigma] = \frac{\E[|X\cap B|]}{L^{2r}}
$$
As such,
$$
L^{2r} \frac{\E[\sigma]}{\E[|X|]} = \frac{\E[|X\cap B|]}{\E[|X|]} = 1 - \frac{\E[|X\backslash B|]}{\E[|X|]}
$$
Since \cite{sawyer1979limit} showed that $0\leq |X|/\E[|X|]$ converges in distribution to a unit exponential random variable and its mean converges to the mean of the limiting distribution, then $|X|/\E[|X|$ is uniform integrable (UI). As $0\leq |X\backslash B|/\E[|X|] \leq |X|/\E[|X|]$, then $|X\backslash B|/\E[|X|]$ is also UI. Using $|X\backslash B|\rightarrow 0$ in probability by \eqref{eqn:inBox}, it then follows that $\E[|X\backslash B|]/\E[|X|] \rightarrow 0$ as $L\rightarrow \infty$, yielding the desired result.

Lastly, we prove \eqref{eqn:Sprob} by using Chebychev's inequality, which states
$$
P\left(\left|\frac{\sigma}{\E[\sigma]}-1\right| > \epsilon\right) \leq \frac{\V[\sigma]}{\epsilon^2 \E[\sigma]^2}.
$$
As such, we seek to upper bound $\E[\sigma^2]$ to upper bound $\V[\sigma]$. We note
$$
\sigma^2 = \sum_{i, j} p_i^2p_j^2 \Rightarrow \frac{\E[\sigma^2]}{\E[\sigma]^2} = \frac{\E[|X_i \cap B||X_j \cap B|]}{L^{4r} \E[\sigma]^2} 
$$
where $X_i, X_j$ are the sites that share the same species as the uniform randomly chosen sites $x_i, x_j$ respectively. It follows from \eqref{eqn:ES} that
$$
\lim_{L\rightarrow \infty} \frac{\E[\sigma^2]}{\E[\sigma]^2} = \lim_{L\rightarrow \infty} \frac{\E[|X_i \cap B||X_j \cap B|]}{\E[|X|]^2} \leq \lim_{L\rightarrow \infty} \frac{\E[|X_i||X_j|]}{\E[|X|]^2}
$$
Intuitively, this limit is $1$ since $x_i$ and $x_j$ are with high probability very far apart, making $X_i$ and $X_j$ roughly independent. More rigorously, we upper bound $\E[|X_i||X_j|]$ by examining the case where $x_i, x_j$ are close to each other and when they are farther apart. Let $B_i$ and $B_j$ denote the sites that are within $L^{1+\epsilon}$ steps of $x_i$ and $x_j$, respectively. Then
\begin{align*}
    \E[|X_i||X_j|] &= \E[|X_i||X_j| \mathbbm{1}_{(X_i \subset B_i)\cap (X_j\subset B_j) \cap (B_i \cap B_j = \varnothing)}]\\
    &\phantom{=}+\E[|X_i||X_j| \mathbbm{1}_{(X_i \nsubseteq B_i)\cup (X_j\nsubseteq B_j) \cup (B_i \cap B_j \neq \varnothing)}]
\end{align*}

For the first term, we note
\begin{align}
    \E[|X_i||X_j| &\mathbbm{1}_{(X_i \subset B_i)\cap (X_j\subset B_j) \cap (B_i \cap B_j = \varnothing)}]\nonumber\\
    &\leq \E[|X_i||X_j| \mathbbm{1}_{(X_i \subset B_i)\cap (X_j\subset B_j)}: B_i \cap B_j = \varnothing]\nonumber\\
    &= \sum_{x_i' \in B_i} \sum_{x_j' \in B_j} P(x_i' \in X_i, X_i \subset B_i, x_j' \in X_j, X_j \subset B_j : B_i \cap B_j = \varnothing) \nonumber\\
    &\leq \sum_{x_i' \in B_i} \sum_{x_j' \in B_j} P(x_i' \in X_i)P(x_j' \in X_j) \label{eqn:independence}\\
    &= \E[|X_i \cap B_i|]\E[|X_j \cap B_j|] \leq \E[|X|]^2\nonumber
\end{align}
\eqref{eqn:independence} follows from a coupling argument. Consider a CRW with mutations with only two initial walkers on $x_i$ and $x_i'$, which we call CRW1. Then $P(x_i'\in X_i)$ is the probability the two walkers meet before a mutation occurs. We similarly construct another CRW that we call CRW2 with mutation with initial walkers on $x_j$ and $x_j'$ and similarly interpret $P(x_j'\in X_j)$. Now, construct a CRW with mutations whose walkers starting at $x_i, x_j, x_i',$ and $x_j'$ are coupled to the walkers in CRW1 and CRW2 until either the walkers have mutated or one of the walkers starting in $\{x_i, x_i'\}$ meet one of the walkers starting in $\{x_j, x_j'\}$.  Call this CRW3. By construction, if $x_i' \in X_i, X_i \subset B_i, x_j' \in X_j, X_j \subset B_j$ when $B_i \cap B_j = \varnothing$ in CRW3, then the walkers meet in CRW1 and the walkers meet in CRW2. As CRW1 and CRW2 are independent, we have \eqref{eqn:independence}.

For the second term, note that by repeated application of Holder's inequality,
\begin{align*}
    \E[|X_i||X_j| &\mathbbm{1}_{(X_i \nsubseteq B_i)\cup (X_j\nsubseteq B_j) \cup (B_i \cap B_j \neq \varnothing)}]\\
    &\leq \sqrt{\E[|X_i|^2|X_j|^2]P((X_i \nsubseteq B_i)\cup (X_j\nsubseteq B_j) \cup (B_i \cap B_j \neq \varnothing))}\\
    &\leq (\E[|X_i|^4]\E[|X_j|^4])^{1/4}\sqrt{P((X_i \nsubseteq B_i)\cup (X_j\nsubseteq B_j) \cup (B_i \cap B_j \neq \varnothing))}\\
    &\leq \sqrt{\E[|X|^4]}\sqrt{P((X_i \nsubseteq B_i)\cup (X_j\nsubseteq B_j) \cup (B_i \cap B_j \neq \varnothing))}
\end{align*}
Recalling from Lemma \ref{lem:inBox} that the probability goes to $0$ and from \cite{sawyer1979limit} that $\E[|X|^4]/\E[|X|]^4 \rightarrow 4!$ as $L\rightarrow \infty$, it follows that the second term divided by $\E[|X|]^2$ approaches $0$. Therefore,
$$
0\leq \lim_{L\rightarrow \infty} \frac{\V[\sigma^2]}{\E[\sigma]^2} = \lim_{L\rightarrow \infty} \frac{\E[\sigma^2]}{\E[\sigma]^2} - 1\leq \lim_{L\rightarrow \infty} \frac{\E[|X_i||X_j|]}{\E[|X|]^2}-1 =0
$$
yielding our desired result.

\subsection*{Conjecture on a Shannon counterpart to Theorem \ref{thm:Simpson}}
By definition, $\E[H]$ is negative the expected log of the fraction of sites that share the same species as a uniform randomly chosen site, ie
$$
\E[H] = -\E\bigg[\ln\bigg(\frac{|X \cap B|}{L^{2r}}\bigg)\bigg] \approx -\E[\ln(|X|)] +2r\ln(L) = -\E\bigg[\ln\bigg(\frac{|X|}{\E[|X|]}\bigg)\bigg] - \ln(\E[|X|]) +2r\ln(L)
$$
when $r>1$. Since \cite{sawyer1979limit} showed $|X|/\E[|X|]$ converges in distribution to a unit exponential as $L\rightarrow \infty$, we might expect (though we would need uniform integrability to make this rigorous)
$$
\E\bigg[\ln\bigg(\frac{|X|}{\E[|X|]}\bigg)\bigg] \approx \int_0^\infty \ln(s)e^{-s}ds = -\gamma
$$
where $\gamma \approx 0.5772$ is the Euler–Mascheroni constant. If we further assume $H-\E[H]$ converges in probability to $0$, then
\begin{conjecture}
    As $L\rightarrow \infty$,
    $$
    H +\ln(\E[|X|])-2r\ln(L) \rightarrow \gamma
    $$
    in probability.
\end{conjecture}
If this conjecture is true, we can similarly estimate a power law exponent for ${}^1D$ of
$$
\frac{\ln(\ln(L)) + \ln(2/\pi)+\gamma}{2\ln(L)}
$$

\printbibliography

\end{document}